\documentclass[cmp]{svjour}
\usepackage{latexsym}
\usepackage{times}
\usepackage{amsmath}
\usepackage{amssymb}
\usepackage{amscd}
\usepackage{amsfonts}
\usepackage{amstext}
\usepackage{graphicx}
\usepackage[all]{xy}
\usepackage{float}
\newcommand{\hq }{{/\kern -.185em/}}
\def\dim{\operatorname{dim}}
\def\exp{\operatorname{exp}}

\newcommand{\bk}[2]{\ensuremath{\langle #1 | #2 \rangle}}

\newcommand{\Ad}{\mbox{Ad}}
\begin{document}
\title{Bipartite entanglement, spherical actions and geometry of local unitary orbits}
\author{Alan Huckleberry\inst{1} \and Marek
Ku\'s\inst{2} \and Adam Sawicki \inst{2,} \inst{3}
}                     
\institute{Fakult\"at f\"ur Mathematik, Ruhr-Universit\"at Bochum, D-44780
Bochum, Germany \and Center for Theoretical Physics, Polish Academy of
Sciences, Al. Lotnik\'ow 32/46, 02-668 Warszawa, Poland \and School of Mathematics, University of Bristol, University Walk, Bristol BS8 1TW, UK}
\date{Received: date / Accepted: date}
%
\communicated{name}
\maketitle
\begin{abstract}
We use the geometry of the moment map to investigate properties of pure
entangled states of composite quantum systems. The orbits of equally
entangled states are mapped by the moment map onto coadjoint orbits of local
transformations (unitary transformations which do not change entanglement).
Thus the geometry of coadjoint orbits provides a partial classification of
different entanglement classes. To achieve the full classification a further
study of fibers of the moment map is needed. We show how this can be done
effectively in the case of the bipartite entanglement by employing Brion's
theorem. In particular, we presented the exact description of the partial
symplectic structure of all local orbits for two bosons, fermions and
distinguishable particles.
\end{abstract}
\section{Introduction. Statement of results}
\label{sec:intro}

Quantum correlations among parts of a multicomponent (multipartite) composite
quantum systems are usually described and investigated employing tools of
multilinear algebra. This seems to be a natural method since at the most
elementary level pure nonentangled states are simple tensors in the tensor
product of Hilbert spaces of subsystems of an $L$-partite system,
\[
\mathcal{H}=\mathcal{H}_{1}\otimes\cdots\mathcal{H}_{L},
\]
 whereas various degrees of entanglement can be characterized by quantifying
(in some precisely defined sense) the ``departure'' of a state from the
simple tensor structure.

Behind these algebraic definitions one finds interesting geometry. In
\cite{sawicki11} and \cite{sawicki11a} we showed how some methods of
symplectic geometry can be used to describe and quantify entanglement. Two
basic observations are crucial. From a physical point of view entanglement
properties remain unchanged under the action of ``local'' unitary group
$K=U(\mathcal{H}_{1})\times\ldots\times U(\mathcal{H}_{L})$ which reflects
the physically obvious statement that entanglement can not be changed by
non-dissipative quantum operations confined to subsystems. It follows thus
that a classification and description of orbits of $K$ in the space of states
(which, as quantum mechanics demands, is not $\mathcal{H}$ itself, but rather
its projectivization $\mathbb{P}(\mathcal{H})$) leads to a classification of
various degrees of entanglement. The second ingredient of this approach is an
observation that both the space of states $\mathbb{P}(\mathcal{H})$ and the
dual of the Lie algebra $\mathfrak{k}^{\ast}$ of $K$ possesses natural
Poisson structures intertwined by the $K$-equivariant moment map $\mu$ which
sends $K$-orbits in $\mathbb{P}(\mathcal{H})$ onto coadjoint orbits.
Coadjoint orbits are in fact symplectic manifolds but since the moment map is
not diffeomorphic the symplectic form on $\mathbb{P}(\mathcal{H})$ is, in
general, degenerate on $K$-orbits. It is not the case only for the unique
orbit of non-entangled states and the degree of degeneracy of the symplectic
form on an orbit can be used as a measure of entanglement \cite{sawicki11}.
The moment map can be used as a tool for a partial classification of
$K$-orbits in $\mathbb{P}(\mathcal{H})$ since it maps them onto coadjoint
orbits which can be easily classified by their intersections with the dual to
the Cartan subalgebra of the Lie algebra $\mathfrak{k}$ of $K$. Such a
classification would give an effective answer to the problem of the ``local
unitary equivalence of states'' (see \cite{kraus10}), i.e.\ whether two
states can be mutually connected by non-dissipative quantum operations
restricted to subsystems. The problem, important from the experimental point
of view when the possibility of obtaining requested states from those which
are available, translates thus into deciding if two states belong to the same
$K$-orbit. Obviously, since the moment map is not a diffeomorphism one does
not obtain a satisfactory classification of $K$-orbits in
$\mathbb{P}(\mathcal{H})$ by mapping them into coadjoint ones. It may happen
that many (in fact a continuous family of) orbits are mapped onto a single
coadjoint orbit. Hence, for a full characterization a deeper insight into the
structure of fibers of the moment map is needed. We will show how to do it
for $L=2$. We will also show on examples how the structure of fibers and,
consequently, classification of orbits become more complicated for systems
with more parts.

The above outlined formulation of the problem allows for generalizations to
cases when the underlying Hilbert space $\mathcal{H}$ does not have a
structure of the full tensor product, but eg., is its symmetric,
$S^{2}\mathcal{H}_{1}$, or antisymmetric, $\bigwedge^{2}\mathcal{H}_{1}$,
part as in the cases of two indistinguishable particles (bosons and
fermions). Since all three cases bear significant similarities we will
consider them in parallel whenever possible.

Our results are stated precisely in the last section of the paper. For those
readers most interested in the qualitative aspects of these results we state
them here without technical details of proofs.

The goal of our work is to give an exact description of the partial
symplectic structure of all $K$-orbits in $M=\mathbb{P}(\mathcal{H})$ and $L=2$. To do
this we first observe that it follows from Brion's Theorem \cite{brion87}
that the moment map $\mu:M\to\mathfrak{k}^{*}$ parameterizes the $K$-orbits
in $M$ in the sense that it maps the set of $K$-orbits in $M$ bijectively
onto the set of $K$-orbits in its image. It is well-known that, modulo the
action of the Weyl-group, this image is $K.P$ where $P$ a convex region in
$\mathfrak{t}^{*}$. Here $\mathfrak{t}$ is the Lie algebra of diagonal
matrices in $\mathfrak{k}=\mathfrak{su}_{N}$ and $\mathfrak{t}^{*}$ is
embedded in $\mathfrak{k}$ via the invariant form $\langle
A,B\rangle=\mathrm{Tr}(AB)$. In fact $P$ intersects each orbit of $K$ in an
orbit of the Weyl group so that, modulo the Weyl group, $P$ parameterizes the
orbits. We give an exact description of $P$ as a certain probability
polyhedron.

We also determine a real algebraic set $\Sigma_{\mathbb{R}}^{+}$ in $M$,
defined by linear algebraic equations and inequalities, which parameterizes
the $K$-orbits in $M$ and which is mapped onto a fundamental region of the
Weyl-group in $P$ by the moment map. Every element $x$ of
$\Sigma_{\mathbb{R}}^{+}$ determines in a simple way a vector
$d=(d_{1},\ldots,d_{k})$ of positive integers which completely determines the
moment map image $\mu(x)$ as a flag manifold $F(d_{1},\ldots,d_{k})$. We also
exactly describe the fiber $\mathcal{F}_{x}=\mu^{-1}(\mu(x))$ of the moment
map. It is the fiber of the homogeneous fibration $K/K_{x}\to K/K_{\mu(x)}$
which in fact (up to very simple finite-coverings) is a product of certain
symmetric spaces. In the case of bosons it is the product of a torus and a
number (depending on $d$ and the degeneracy) of symmetric spaces of the form
$\mathrm{SU}_{m}/\mathrm{SO}_{m}$. The case of fermions is analogous, except
that the symmetric spaces are of the form $\mathrm{SU}_{m}/\mathrm{USp}_{m}$.

To make the paper reasonably concise and, simultaneously, accessible to
readers not familiar with the whole needed background material we provide in
Appendix comprehensive outline of the most important results concerning
spherical varieties. On the other hand we decided that the main text of the
paper should remain self-contained hence the needed definitions and
statements, even if they concern the background material, are often
accompanied by shortened explanations and arguments fuller versions of which
are given in the Appendix.

\section{Local unitary actions in spaces of states}
\label{sec:2}

Let $\mathcal{H}_{1}\cong\mathbb{C}^{N}$ , the $N$-dimensional complex vector
space equipped with a scalar product $\bk{\cdot}{\cdot}$ be the one particle
Hilbert space. Let $G=\mathrm{SL}_{\mathbb{C}}(\mathcal{H}_{1})$ be the
special linear group and and $K=SU(\mathcal{H}_{1})$ be the subgroup of
unitary transformations in it. The two particle Hilbert spaces for bosons,
fermions and distinguishable particles are, respectively, the symmetric and
antisymmetric part of the tensor product of two copies of $\mathcal{H}_{1}$
and the full tensor product itself,
\begin{eqnarray}
\mathcal{H}_{B} & = & S^{2}\mathcal{H}_{1},\,\,\mathrm{bosons}\nonumber \\
\mathcal{H}_{F} & = & \bigwedge^{2}\mathcal{H}_{1},\,\,\mbox{fermions}\label{eq:spaces}\\
\mathcal{H}_{D} & = & \bigotimes^{2}\mathcal{H}_{1}=\mathcal{H}_{B}\oplus\mathcal{H}_{F},\,\,\mbox{distinguishable particles}\nonumber
\end{eqnarray}

We have the natural action of $G_{D}=G\times G$ and $K_{D}=K\times K$ on
$\mathcal{H}_{D}$ given by
\begin{gather*}
(U_{1},U_{2}).(v\otimes w)=U_{1}\otimes U_{2}(v\otimes w)=(U_{1}v\otimes
U_{2}w).
\end{gather*}
 The diagonal action of $G$ and $K$ on $\mathcal{H}_{D}$, i.e.
\begin{gather*}
U.(v\otimes w)=U\otimes U(v\otimes w)=Uv\otimes Uw,
\end{gather*}
 induces the action of $G$, $K$ on $\mathcal{H}_{B}$ and $\mathcal{H}_{F}$.
In the following we will denote by $v\vee w:=v\otimes w+w\otimes v$ and
$v\wedge w:=v\otimes w-w\otimes v$ the symmetric and antisymmetric tensor
products. It will be convenient to regard the tensors at hand as matrices.
Therefore we let $\{e_{1},\,\ldots,e_{N}\}$ be an orthonormal basis of
$\mathcal{H}_{1}$, define $H$ (the complex torus) to be the subgroup of $G$
of diagonal matrices and let $T:=H\cap K$ be the corresponding real torus
consisting of unitary diagonal matrices with determinant equal to one.
Observe that $\{e_{ij}^{s}:=e_{i}\vee e_{j};\;1\le i,j\le N\}$ and
$\{e_{ij}^{a}:=e_{i}\wedge e_{j};\; i<j\}$ are orthogonal bases of
$\mathcal{H}_{B}$ and $\mathcal{H}_{F}$ respectively. We can regard any
tensor $v$ in $\mathcal{H}_{D}$ as $N\times N$-matrix $C_{v}$.
\begin{equation}
v=\sum_{i,j}(C_{v})_{ij}e_{i}\otimes e_{j}.\label{}
\end{equation}
 The matrix $C_{v}$ is symmetric for symmetric tensors and antisymmetric
for antisymmetric ones hence we can write analogous formulae for
$v\in\mathcal{H}_{B}$ and $v\in\mathcal{H}_{B}$ substituting $e_{ij}^{s}$ and
$e_{ij}^{a}$ in place of $e_{i}\otimes e_{j}$ completing thus the
identification of tensors with matrices in all three considered cases.

The action of $G$ and $G_{D}$ and hence of $K$ and $K_{D}$ translated into
language of matrices is given by
\begin{eqnarray}
U.C_{v} & = & UC_{v}U^{t},\,\,\mbox{bosons and fermions,}\label{eq:matrix_action}\\
(U,V).C_{v} & = & UC_{v}V^{t},\,\,\mbox{distinguishable particles},\nonumber
\end{eqnarray}
 where $^{t}$ denotes the transposition. Unless there is a danger
of confusion we will let $\mathcal{H}$ denote any of the vector spaces
(\ref{eq:spaces}) and by $M=\mathbb{P}(\mathcal{H})$ the associated complex
projective space. There are of course differences between these three cases,
but conceptually speaking, these are slight. In order to facilitate a
simultaneous treatment we let $n=N$ in the case of $S^{2}V$ and $2n=N$ (resp.
$2n+1=N$) in the case of $\Lambda^{2}V$ where $N$ is even (resp. odd).
Finally we will denote
\begin{gather*}
s_{N}=e_{1}\vee e_{1}+e_{2}\vee e_{2}+\ldots+e_{N}\vee e_{N},\\
a_{N}=e_{1}\wedge e_{2}+e_{3}\wedge e_{4}+\ldots+e_{2n-1}\wedge e_{2n},\\
d_{N}=e_{1}\otimes e_{1}+e_{2}\otimes e_{2}+\ldots+e_{N}\otimes e_{N}.
\end{gather*}
 We will use $x_{N}$ to denote any of these points when treating
the corresponding case. Notice that the matrices for these tensors are $C_{s_{N}}=2I$, $C_{d_{N}}=I$,
and $C_{a_{N}}=J$ where $J$ is a block diagonal matrix each $2\times2$ block
being standard symplectic matrix, i.e.,

\begin{gather*}
J=\left(\begin{array}{cc}
0 & 1\\
-1 & 0
\end{array}\right)\dotplus\ldots\dotplus\left(\begin{array}{cc}
0 & 1\\
-1 & 0
\end{array}\right).
\end{gather*}

\section{The spherical property}
\label{sec:spherical}

A complex Lie group $G$ of matrices is said to be \emph{reductive} if it is
the complexification $G=K^{\mathbb{C}}$ of its maximal compact subgroup. In the
case of interest here $G=\mathrm{SL}_{N}(\mathbb{C})$ is the complexification
of the compact subgroup $K=\mathrm{SU}(N)$. Let $H$ be an algebraic subgroup
of $G$ and denote by $\Omega:=G/H$ corresponding homogenous space.

\begin{definition}A $G$-homogenous spece $\Omega=G/H$ is said to be a
spherical homogenous space if and only if some (and therefore every) Borel
subgroup $B\subset G$ has an open dense orbit in $\Omega$.
\end{definition}

Recall that by definition a Borel subgroup is a maximal connected solvable
subgroup of $G$. In our particular case of interest where
$G=\mathrm{SL}_{N}(\mathbb{C})$, the group of upper-triangular matrices,
which can be regarded as the stabilizer of the standard full flag
\[
0\subset\mathrm{Span}\{e_{1}\}\subset\mathrm{Span}\{e_{1},e_{2}\}\subset\ldots\subset\mathrm{Span}\{e_{1},\ldots e_{N-1}\}\subset\mathcal{H}\,,
\]
 is a good example of a Borel subgroup. In general, every two Borel
subgroups are conjugate by an element of $G$. Hence in our case $B$ is a
Borel subgroup if and only if it is the stabilizer of some full flag.

Suppose now that $M$ is an irreducible algebraic variety (for example a
complex projective space $\mathbb{P}(\mathcal{H})$) equipped with a
$G$-action.


\begin{definition}If $G$ has an open (Zariski dense) orbit $\Omega=G/H$
in $M$ and $\Omega$ is a spherical homogenous space, then $M$ is said to be a
spherical embedding of $\Omega$.
\end{definition}

In the next sections we show that spherical homogenous spaces possess a
number of crucial properties which we will use to analyze the geometry of
entangled states. Let us start with presenting a basic example of such
spaces, vis.\ affine symmetric spaces which will play an important role in
the due course.

\paragraph*{Affine symmetric spaces}

Let $\mathfrak{g}$ be a complex semisimple Lie algebra,
$\theta:\mathfrak{g}\to\mathfrak{g}$ a complex linear involution, i.e., a Lie
algebra automorphism with $\theta^{2}=\mathrm{Id}$, and
$\mathfrak{g}=\mathfrak{h}\oplus\mathfrak{p}$ the decomposition of
$\mathfrak{g}$ into $\theta$-eigenspaces. The fixed point algebra
$\mathfrak{h}$, i.e., the subspace belonging to the eigenvalue $+1$ of
$\theta$, defines a (reductive - see Appendix)) closed subgroup $H$ in $G$. The
complex homogenous space $G/H$ is in this case an \textit{affine symmetric
space}. In the Appendix we prove
\begin{proposition}\label{symmetric} Affine symmetric
spaces are spherical.
\end{proposition}

Our goal now is to prove that $M=\mathbb{P}(\mathcal{H})$, where
$\mathcal{H}$ is the Hilbert space for two fermions, bosons or
distinguishable particles, is a spherical embedding of some open dense orbit
$\Omega$
of $G$ action%
\footnote{The respective $G$ actions were defined in Section~\ref{sec:2}%
}. We will first prove

\begin{proposition} The orbits $G.s_{N}$, $G.a_{N}$, $G_{D}.d_{N}$
are affine symmetric spaces.
\end{proposition}

\begin{proof} The isotropy subgroups (stabilizers) of the points
$s_{N}$, $a_{N}$, and $d_{N}$ under the corresponding actions of $G$ or
$G_{D}$ are given by $G_{s_{N}}=\{T\in G:\,
TM_{s_{N}}T^{t}=M_{s_{N}}\}=\{T\in G:\, TT^{t}=I\}$, $G_{a_{N}}=\{T\in G:\,
TM_{a_{N}}T^{t}=M_{a_{N}}\}=\{T\in G:\, TJT^{t}=J\}$, and
$G_{Dd_{N}}=\{(T,S)\in G\times G:\, TM_{d_{N}}S^{t}=M_{d_{N}}\}=\{(T,S)\in
G\times G:\, TS^{t}=I\}$, where we use the explicit forms of the matrices
$M_{s_{N}}$, $M_{a_{N}}$, and $M_{d_{N}}$ given at the end of
Section~\ref{sec:2}. On the other hand they are given as fixed point sets of
the holomorphic involutions $\theta(T)=(T^{t})^{-1}$,
$\theta(T)=J(T^{t})^{-1}J^{-1}$, and
$\theta(T,S)=\left((S^{t})^{-1},(T^{t})^{-1}\right)$,
respectively%
\footnote{On the level of the algebras the corresponding involutions (denoted
by the same letter $\theta$) read, respectively $\theta(X)=-X^{t}$,
$\theta(X)=-JX^{t}J^{-1}$, and $\theta(X,Y)=(-X^{t},-Y^{t})$%
}. Hence in all three cases cases the orbits $G.s_{N}=G/G_{s_{N}}$,
$G.a_{N}=G/G_{a_{N}}$, and $G_{D}.d_{N}=G_{D}/G_{Dd_{N}}$ are affine
symmetric spaces.\end{proof}

\begin{proposition} The orbits $G.s_{N}$, $G.a_{N}$, $G_{D}.d_{N}$
are open Zariski dense in the appropriate $\mathbb{P}(\mathcal{H})$.
\end{proposition}

\begin{proof}Treated as quadratic forms the points $s_{N}$, $a_{N}$,
$d_{N}$ are of maximal rank. Since the rank of such a form is the only
invariant in the sense that any two forms of the same rank are equivalent
modulo the standard $\mathrm{GL}(\mathcal{H})$-action, it follows that
$\mathrm{GL}(\mathcal{H}).x_{N}$ is dense and open in $\mathcal{H}$ and
$\mathrm{SL}(\mathcal{H}).x_{N}=G.x_{N}$ is open in
$M=\mathbb{P}(\mathcal{H})$ \end{proof}

\begin{proposition}\label{open_spherical} $M=\mathbb{P}(\mathcal{H})$,
where $\mathcal{H}$ is the Hilbert space for two fermions, bosons or
distinguishable particles is spherical embedding of $G.s_{N}$, $G.a_{N}$,
$G_{D}.d_{N}$ respectively
\end{proposition}
\begin{proof}Follows from the above propositions.\end{proof}

\section{Moment map}
\label{sec:moment}

In the following we show that it is very useful to view spherical varieties
from the standpoint of symplectic geometry. To this end we devote this
section to the concept of the moment map (see \cite{guillemin84} for more
details). We discuss it in general as well as in the specific context of our
problem. Finally we derive an explicit formula for the moment map associated
to the unitary action on the complex projective space.

\subsection{General setting}
\label{subsec:momentgeneral}

Let $K$ be any connected, semisimple Lie group acting smoothly as group of
symplectomorphisms of a symplectic manifold $(M,\omega)$, i.e.

\begin{gather*}
K\times M\ni(g,x)\mapsto\Phi_{g}(x)\in M,\quad\quad
\Phi_{g}^{*}\omega=\omega.
\end{gather*}
 For any $\xi\in\mathfrak{k}$ we define a fundamental vector field
$\hat{\xi}$ for the $\Phi$ action
\begin{equation}
\hat{\xi}(x)=\frac{d}{dt}\bigg|_{t=0}\Phi_{\exp t\xi}(x).\label{fvf}
\end{equation}
 In the following, to shorten the notation we will use $gx$ or $g.x$
instead of $\Phi_{g}(x)$ for $g\in G$, $x\in K$.

The map
\begin{gather}\label{eq:fvf1}
\widehat{}:(\mathfrak{k},[\cdot,\cdot])\rightarrow(\chi(M)\,,\,[\cdot,\cdot]),
\quad \xi\mapsto\hat{\xi},
\end{gather}
is a homomorphism of the Lie algebra $\mathfrak{k}$ into the Lie algebra
$\chi(M)$ of the vector fields on $M$.

A vector field $X$ is said to be Hamiltonian if the Lie derivative of $\omega$
along it vanishes,
\begin{gather}
X\in \mathrm{Ham}(M)\equiv \mathcal{L}_H\omega=0.\label{eq:hamvf}
\end{gather}
>From the Cartan formula
\begin{gather}
\mathcal{L}_X\omega= i_X\circ d\omega+d\circ i_X\omega.\label{eq:cartan1}
\end{gather}
and $d\omega=0$ it follows that $d\circ i_X\omega=d\omega(X,\cdot)=0$ for a
Hamiltonian vector field $X$, i.e. the 1-form $\omega(X,\cdot)$ is closed. If
the first de Rham cohomology group of $M$ vanishes then all closed forms are
exact, i.e., there exists a function $F$ such that $\omega(X,\cdot)=dF$. In
this case the map $F\mapsto X_F$ which which assigns a Hamiltonian vector
field $X_{F}\in\mathrm{Ham}(M)$ to any smooth function $F\in C^\infty(M)$
\emph{via} $\omega(X_F,\cdot)=dF$ is surjective. Moreover, one checks easily
that if we define the Poisson bracket on $M$ in the usual way,
\begin{gather}\label{PB}
\{F\,,\, G\}=\omega(X_{F},X_{G}),\,\,\,\, F,G\in C^\infty(M),
\end{gather}
the map
\begin{gather}\label{ham}
\rho:\big(C^\infty(M),\{\cdot,\cdot\}\big)
\rightarrow\big(\mathrm{Ham}(M),[\cdot,\cdot]\big), \quad F\mapsto X_F,
\end{gather}
is a homomorphism of Lie algebras.

Observe now that since the action of $K$ on $M$ is symplectic we have
$\mathcal{L}_{\hat{\xi}}\omega=0$ for the fundamental vector fields
(\ref{fvf}), i.e., the fundamental vector fields are Hamiltonian. Hence, under
the assumption that the first de Rham cohomology group of $M$ vanishes, for
each Lie algebra element $\xi$ there exists a function $\mu_\xi$ such that
\begin{gather}\label{eq:hamiltonian}
d\mu_{\xi}=\omega(\hat{\xi},\cdot).
\end{gather}
where $\xi$ is given by (\ref{fvf}). Notice that for any $\xi$ the function
$\mu_{\xi}$ is given only up to an arbitrary constant function. Since
$\mathfrak{k}$ is finite dimensional, $\mu_{\xi}$ can be chosen to be linear in
$\xi$, i.e.,
\begin{equation}
\mu_{\xi}(x)=\langle\mu(x),\xi\rangle,\quad\mu(x)\in\mathfrak{k}^{\ast},\label{moment-lin}
\end{equation}
where $\langle\,,\rangle$ is the pairing between $\mathfrak{k}$ and
$\mathfrak{k}^{\ast}$. In this way we obtain a map
$\mu:M\rightarrow\mathfrak{k}^{\ast}$ called the moment map.

Without spoiling the properties of the moment map which we have just derived
we still have the freedom to add a constant function to $\mu_\xi$ which is
equivalent to adding a linear functional $\nu\in\mathfrak{k}^\ast$ to the
mapping $\bar\mu$ from the Lie algebra $\mathfrak{k}$ to the space of smooth
functions
\begin{equation}\label{eq:moment-hom}
\mathfrak{k}\ni\xi\mapsto \bar\mu(\xi):=\mu_\xi\in C^\infty(M).
\end{equation}
It can be shown that under a particular assumption about the group $K$ which
is always fulfilled through the paper, \textit{vis}., its semisimplicity, we
can use this freedom to make $\bar{\mu}^\prime=\bar{\mu}+\nu$ a homomorphism
of the Lie algebras $\mathfrak{k}$ and $C^\infty(M)$, the latter again
endowed with the Lie algebra structure with the help of the Poisson bracket
(\ref{PB}). In fact the choice can be made in a unique way such that
$\bar{\mu}^\prime\circ\rho=\widehat{}$\ . A short proof of the above
statement can be found in \cite{guillemin84}. Here we remark only that if the
second cohomology group of the Lie algebra $\mathfrak{k}$  vanishes,
$H^{2}(\mathfrak{k},\mathbb{R})=0$, then a $\nu$ making $\bar{\mu}$ the
homomorphism always exists and if the first cohomology group vanishes as
well, $H^{1}(\mathfrak{k},\mathbb{R})=0$, then the choice of such $\nu$, and
consequently of $\bar{\mu}^\prime$, is unique. The structure of cohomology
groups of semisimple Lie algebras is characterized by Whitehead's lemmas,
which state that, in particular for semisimple algebras,
$H^{1}(\mathfrak{k},\mathbb{R})=0=H^{2}(\mathfrak{k},\mathbb{R})$. In the
following we assume that our moment map $\mu$ is such that $\bar{\mu}$ is the
above described unique homomorphism.

The group $K$ acts also on its Lie algebra $\mathfrak{k}$ \textit{via} the
adjoint action ${\Ad}_{g}\xi=g\xi g^{-1}$. The dual of the adjoint action is
the coadjoint action on $\mathfrak{k}^{\ast}$,
\begin{equation}
\langle{\Ad}_{g}^{\ast}\alpha,\xi\rangle=\langle\alpha,{\Ad}_{g^{-1}}\xi\rangle=\langle\alpha,g^{-1}\xi g\rangle,\label{Adast}
\end{equation}
for $g\in K$, $\xi\in\mathfrak{k}$, and $\alpha\in\mathfrak{k}^{\ast}$. Since
$\bar\mu$ is the unique homomorphism the moment map is equivariant (for
details, consult again \cite{guillemin84}), i.e., for each $x\in M$ and $g\in
K$,
\begin{equation}\label{moment-equiv}
\mu\left(\Phi_{g}(x)\right)={\Ad}_{g}^{\ast}\mu(x)
\end{equation}

On each coadjoint orbit, $\Omega_{\alpha}=\{{\Ad}_{g}^{\ast}\alpha:g\in G\}$,
there is a canonical symplectic structure - the so called
Kirillov-Kostant-Souriau form given by
\begin{equation}
\label{symformcoa}
\tilde{\omega}_{\alpha}(\tilde{\xi},\tilde{\eta})
=\langle\alpha,[\xi,\eta]\rangle,
\end{equation}
where $\tilde{\xi},\,\tilde{\eta}$ are the fundamental vector fields
associated to $\xi$, $\eta$ by the coadjoint action. Finally let us notice
that using (\ref{eq:hamiltonian}) we get
\begin{equation}
d\mu_{\xi}(x)(\widehat{\eta})=\omega(x)(\widehat{\xi},\widehat{\eta}),\label{basiclemma}
\end{equation}
and
\begin{proposition} For every $x\in M$ it follows that
\[
\mathrm{Ker}(d\mu(x))=(T_{x}K.x)^{\perp_{\omega}}\,,
\]
\end{proposition}
where by $^{\perp_{\omega}}$ we denote the $\omega$-orthogonal complement

Let us denote by $\mathcal{F}_{x}:=\mu^{-1}(\mu(x))$ the $\mu$ fiber over
$x\in M$. In general little can be said about the position of $\mathcal{F}_x$
with respect to $K.x$. Even if it is smooth we have only
$T_{x}\mathcal{F}_{x}\subset(T_{x}K.x)^{\perp_{\omega}}$ which implies that
$\mathcal{F}_{x}$ may not be entirely contained in the orbit $K.x$. However,
when $K.x$ is coisotropic we have

\begin{corollary} If $K.x$ is co-isotropic, i.e., if $(T_{x}K.x)^{\perp_{\omega}}\subset T_{x}K.x$,
the connected component of the $\mu$-fiber at $x$ is contained in $K.x$.
\end{corollary}

In order to emphasize the fact that the $\mu$-fiber might not be contained in
$K.x$ let us consider following example.

\begin{example}\label{ex1}Let $\mathcal{H}_{1}=\mathbb{C}^{2}$,
and $\mathcal{H}=\mathcal{H}_{1}^{\otimes3}$, $K=SU_{2}$,
$G=SL_{2}(\mathbb{C})^{\times3}$. Denote by $\{e_{1},\,
e_{2}\}\subset\mathcal{H}_{1}$ the orthogonal basis of
$\mathbb{\mathcal{H}}_{1}$. Consider now two states
\begin{gather*}
x_{1}=\sqrt{\frac{2}{3}}e_{1}\otimes e_{1}\otimes e_{1}+\frac{1}{\sqrt{3}}e_{2}\otimes e_{2}\otimes e_{2},\\
x_{2}=\frac{1}{\sqrt{3}}\left(e_{2}\otimes e_{1}\otimes e_{1}+e_{1}\otimes
e_{2}\otimes e_{1}+e_{1}\otimes e_{1}\otimes e_{2}\right).
\end{gather*}
Using formula (\ref{generalmomentmap}) (see below) one easily checks that
$\mu(x_{1})=\mu(x_{2})$. Notice however that the orbits $K.x_{1}$ and
$K.x_{2}$ are disjoint as the orbits $G.x_{1}$ and $G.x_{2}$ are known to be
disjoint \cite{dur00}.
\end{example}

\subsection{Moment map on a complex projective space}
\label{subsec:momentprojective}

Let us now restrict to a setting of special importance for our purposes,
namely $M=\mathbb{P}(\mathcal{H})$. For this let $K$ be a connected compact
Lie group acting linearly \emph{via} a unitary representation on a Hilbert
space $\mathcal{H}$. Note that the complexification $G:=K^{\mathbb{C}}$ is
holomorphically represented on $\mathcal{H}$. Our goal now is to derive
formula for the moment map. To this end let $\rho:=\log\Vert\cdot\Vert^{2}$
denote the associated ($K$-invariant) norm-squared function and consider the
$(1,1)$-form on $\mathcal{H}\setminus\{0\}$
\begin{gather*}
\omega:=-\frac{i}{2}\partial\overline{\partial}\rho.
\end{gather*}
 For concrete calculations it is sometimes convenient to use local
complex coordinates $z_{k}=x_{k}+iy_{k}$,
\begin{gather*}
\omega=-\frac{i}{2}\sum_{k,l}(\partial_{z_{k}}\partial_{\overline{z}_{l}}\rho)dz_{k}\wedge d\overline{z}_{l},\\
\partial_{z_{k}}=\frac{1}{2}(\partial_{x_{k}}-i\partial_{y_{k}}),\quad\partial_{\overline{z}_{k}}=\frac{1}{2}(\partial_{x_{k}}-i\partial_{y_{k}}),\\
dz_{k}=dx_{k}+idy_{k},\quad d\overline{z}_{k}=dx_{k}-idy_{k}.
\end{gather*}

The form $\omega$ is degenerate along the fibers of the standard projection
$\mathcal{H}\setminus\{0\}\to\mathbb{P}(\mathcal{H})$ and induces a
symplectic form $\omega_{FS}$ on $\mathbb{P}(\mathcal{H})$ called the
Fubini-Study structure associated to the given unitary
structure on $\mathcal{H}$ %
\footnote{Let us recall that the Hermitian structure on $\mathcal{H}$ endows
$\mathbb{P}(\mathcal{H})$ with a real Riemannian structure - the Fubini-Study
metric. Together with $\omega$ it makes $\mathbb{P}(\mathcal{H})$
a K\"ahler manifold%
}. We wish to define a moment map for the symplectic $K$-action on
$(\mathbb{P}(\mathcal{H})\,,\,\omega_{FS})$. For this we note that in general
if $U$ is a Lie group containing $K$ and having a moment map
$\mu_{U}:M\to\mathfrak{u}^{*}$ and $\pi:\mathfrak{u}^{*}\to\mathfrak{k}^{*}$
is the projection which is induced by the inclusion
$\mathfrak{k}\hookrightarrow\mathfrak{g}$, then $\mu_{K}:=\pi\circ\mu_{U}$ is
a moment map for the $K$-action. We now apply this to the case where $K$ is
acting via a unitary representation as above and $U=U(\mathcal{H})$ is the
full unitary group on $\mathcal{H}$.

The moment map for the action of the full unitary group can be obtained in
the following way. For $\xi\in\mathfrak{u}(\mathcal{H})$ define
\begin{equation}
\mu_{\xi}:=\frac{1}{4}J\widehat{\xi}\log||\cdot||^{2}=\frac{1}{4}J\widehat{\xi}\rho,\label{momentformula}
\end{equation}
 where $J$ is the linear complex structure on $\mathcal{H}$ regarded
as a real Hilbert space. It is the matter of straightforward but tedious
calculations to show that in fact,
\begin{gather*}
\omega:=-\frac{i}{2}\partial\overline{\partial}\rho=\frac{1}{4}dd^{c}\rho,
\end{gather*}
 where $d^{c}f(v)=(Jv)(f)=df(Jv)$ for a tangent vector $v\in\mathcal{H}$.
By these definitions we immediately get that
$\mu:\mathcal{H}\to\mathfrak{u}(\mathcal{H})^{\ast}$ has the properties of a
moment map, i.e., $d\mu_{\xi}=\omega(\hat{\xi},\cdot)$ and
$\mu(g(v))=\Ad^\ast_g(\mu(v))$.
 Since $\omega$ is degenerate and factors through
 $\mathcal{H}\setminus\{0\}\to\mathbb{P}(\mathcal{H})$,
i.e., it is constant along the fibers of this projection, one shows easily
that $\mu$ enjoys the same properties and defines a moment map
$\mu:\mathbb{P}(\mathcal{H})\to\mathfrak{u}(\mathcal{H})^{\ast}$ with respect
to the $U(\mathcal{H})$-invariant Fubini-Study form $\omega_{FS}$. Let us
explicitly compute this map. For $\xi\in\mathfrak{u}(\mathcal{H})$ from
(\ref{momentformula}) and the properties of the scalar product
$\bk{\cdot}{\,\cdot}$ we obtain
\begin{equation}
\mu_{\xi}(v)=\frac{1}{4}J\widehat{\xi}\rho(v)=\frac{1}{4}\frac{d}{dt}\Big\vert_{t=0}\log\bk{e^{it\xi}v}{e^{it\xi}v}=\frac{1}{2}\mathrm{Im}\frac{\bk v{\xi v}}{\bk vv}=\frac{i}{2}\frac{\bk v{\xi v}}{\bk vv}\,.\label{generalmomentmap}
\end{equation}
 Direct calculations with the help of (\ref{generalmomentmap}) leads
to yet another very useful formula for $\omega$,
\begin{equation}
\omega(\hat{\xi}_{1},\hat{\xi}_{2})
=-\frac{i\bk v{[\xi_{1}\,,\,\xi_{2}]v}}{2\bk vv},
\quad\xi_{1},\,\xi_{2}\in\mathfrak{u}(\mathcal{H}).\label{ome}
\end{equation}
Using (\ref{generalmomentmap},\ref{ome}) one easily checks that indeed both
$\omega$ and $\mu$ factor through
$\mathcal{H}\setminus\{0\}\to\mathbb{P}(\mathcal{H})$.

The only missing ingredient in the formula (\ref{generalmomentmap}) is an
explicit form for the Hermitian structure on $\mathcal{H}$ treated as space
of (symmetric, antisymmetric, or general) matrices. To avoid computing we use
the irreducibility of the representations of the considered groups on the
corresponding $\mathcal{H}$ in all three cases. The arguments are slightly
different for the diagonal action of
$G=\mathrm{SL}_{\mathbb{C}}(\mathcal{H}_{1})=\mathrm{SL}_{N}(\mathbb{C})$ in
the symmetric and antisymmetric cases one one hand and the case of
$G_{D}=\mathrm{SL}_{N}(\mathbb{C})\times\mathrm{SL}_{N}(\mathbb{C})$ in the
general case for distinguishable particles, so let us start with the former.
\begin{proposition} The $G$-representations on
$\mathcal{H}_{B}$ and $\mathcal{H}_{F}$ are irreducible.
\end{proposition}
\begin{proof} Note that the action of the maximal complex torus $H$
is diagonalized in the bases chosen above for the two cases, i.e., for
$v_{ij}$ a basis element $h(v_{ij})=\chi_{ij}(h)v_{ij}$ for an explicitly
known character $\chi_{ij}$. One checks that no two such characters are the
same, i.e., in both cases the $H$-action is multiplicity-free. Now if $B$ is
the Borel group of upper-triangular matrices which contains $H$, then a
$B$-eigenvector is certainly an $H$-eigenvector. But one checks that among
the $v_{ij}$ the only $B$-eigenvector is $e_{1}\wedge e_{2}$ (resp.
$e_{1}\vee e_{1}$) for $\mathcal{H}_{B}$ (resp. $\mathcal{H}_{F}$).

Now recall that a holomorphic $G$-representation on a complex vector space
$\mathcal{H}$ is irreducible if and only if some (and therefore every) Borel
subgroup $B$ has exactly one fixed point in $\mathbb{P}(\mathcal{H})$, i.e.,
up to complex multiples there is a unique highest weight vector. Since this
is the case in our examples, the desired result has been proved. \end{proof}

The following is the metric version of the criterion which was used above for
the irreducibility of a representation.
\begin{proposition}
A representation of a compact Lie group $K$ as a group of complex linear
transformations on a complex vectors space $\mathcal{H}$ is irreducible if
and only if there is a unique (up to scalar factors) $K$-invariant unitary
structure on $\mathcal{H}$.
\end{proposition}
As a consequence any $K$-invariant Hermitian structure which we choose on
$\mathcal{H}$ will automorphically be the unique one which will define
$\omega_{FS}$ and the associated moment map on $\mathbb{P}(\mathcal{H})$. Let
us now return to regarding a point in $\mathcal{H}$ as a matrix $M$ and let
\[
\langle C_{1},C_{2}\rangle:=\mathrm{Tr}(C_{1}^{\dagger}C_{2})\,.
\]
\begin{corollary} Up to scalar multiples the Hermitian structure
$\langle\ ,\ \rangle$ is the unique $K$-invariant structure on
$\mathcal{H}=\mathcal{H}_{B},\mathcal{H}_{F}$.
\end{corollary}
Using this we get in both cases the following expression for the moment map
\begin{equation}
\mu_{\xi}([v])=\frac{i}{2}\frac{\mathrm{Tr}{(C_{v}^{\dagger}C_{v}\xi)}}{\mathrm{Tr}{(C_{v}^{\dagger}C_{v})}}\,.\label{momentmatrices}
\end{equation}

In the case of distinguishable particles, i.e., for
$\mathcal{H}=\mathcal{H}_{D}\cong\mathbb{C}^{N}\otimes\mathbb{C}^{N}$ we
define a $K$-equivariant antilinear isomorphism $v\mapsto\langle v|\;\rangle$
on the second factor. This defines a symplectomorphism
$\mathbb{C}^{N}\otimes\mathbb{C}^{N}\to\mathbb{C}^{N}\otimes(\mathbb{C}^{N})^{*}$
from the standard symmplectic structure on the tensor product to the
symplectic structure on $\mathrm{End}(\mathbb{C}^{N})$ defined by the unitary
pairing
\[
\langle C_{1},C_{2}\rangle:=\frac{1}{2}\mathrm{Tr}(C_{1}^{\dagger}C_{2})\,.
\]
This identification is not complex linear, in fact the $G$-representation is
not equivalent to the original one. However, since we are only interested in
the symplectic phenomena related to the $K$-representation, this is of no
importance.

As a result of the above we have reduced the problem to considering the $G$-
representation on $W:=\mathrm{End}(\mathbb{C}^{N})$ defined by
$(A,B)(C):=ACB^{-1}$. Due to the fact that the $K$-representation on $W$ is
irreducible, the above unitary structure is (up to a scalar factor) the
unique $K$-invariant structure and the associated moment
map is defined as before by (\ref{momentmatrices}). 

\section{Brion's theorem}
\label{sec:brion}

In Section~\ref{sec:moment} we observed that the $\mu$-fiber may not be
entirely contained in the $K$ orbit. Equivalently the moment map
$\mu:M\rightarrow\mathfrak{k}^{\ast}$ may not separate all $K$ orbits.
Continuing our exposition of symplectic aspects of spherical varieties we
will give in this section conditions ensuring that $\mu$ does separate all
$K$-orbits. This beautiful characterization is due to Brion \cite{brion87}.
Here we give a brief exposition based on \cite{huckleberry90}.

\begin{theorem}\label{brion result}(Brion) Let $K$ be a connected
compact Lie group acting on connected compact K\"{a}hler manifold
$(M,\omega)$ by a Hamiltonian action and let $G=K^{\mathbb{C}}$. The
following are equivalent
\begin{enumerate}
\item $M$ is a spherical embedding of the open $G$-orbit.
\item For every $x\in M$ the $\mu$-fiber $\mu^{-1}(\mu(x))$ is contained
    in $K.x$.
\end{enumerate}
\end{theorem}

Let us elucidate again that $2.$ states exactly that the moment map separates
all $K$-orbits, i.e., from $\mu(x)=\mu(y)$ it follows $y\in K.x$. In this
situation the local unitary equivalence of states $x$ and $y$ can be
effectively checked by establishing whether $\mu(x)$ nad $\mu(y)$ lie on the
same coadjoint orbit of $K$. This in turn can be easily seen by finding
points where coadjoint orbits through $\mu(x)$ and $\mu(y)$ cross the dual to
the maximal commutative subalgebra of $\mathfrak{k}$ (see \cite{sawicki11a}),
or in simple words by diagonalizing $\mu(x)$ and $\mu(y)$ by the coadjoint
action of $K$ and comparing the spectra which should coincide if both
$\mu(x)$ and $\mu(y)$ belong to the same coadjoint orbit.

The Brion's theorem is one of the basic tools we use in our reasoning, hence we
believe it is expedient to sketch its proof.

For the purpose of this paper it is important to understand the proof of
Brion's theorem only in one direction, namely that $1.$ implies $2$. Following
\cite{huckleberry90} there are two crucial ingredients in this proof. The first
one is the so called Marsden-Weinstein reduction which we now briefly describe.

Let $(M,\omega_{M})$ and $(N,\omega_{N})$ be two symplectic manifolds
equipped with a Hamiltonian action of a compact connected Lie group $K$. Let
$\mu_{M}:M\rightarrow\mathfrak{k}^{\ast}$ and
$\mu_{N}:N\rightarrow\mathfrak{k}^{\ast}$ be the corresponding moment maps.
Notice first that $(M\times N,\omega)$, where
\begin{gather}
\omega=\pi_{M}^{\ast}\omega-\pi_{N}^{\ast}\omega,\label{eq:sympdef}
\end{gather}
 and $\pi_{M}:M\times N\rightarrow M$, $\pi_{N}:M\times N\rightarrow N$
are natural projections, is a symplectic manifold. Indeed, the $2$-form
$\omega$ is closed since $d$ commutes with $\pi_{N,M}^{\ast}$. It is also
nondegenerate since assuming to the contrary that there is a vector
$u=(u_{1},u_{2})$ such that for any $v=(v_{1},v_{2})$, $\omega(u,v)=0$,
implies that
\begin{gather*}
\omega(u,v)=\omega_{M}(u_{1},v_{1})-\omega_{N}(u_{2},v_{2})=0.
\end{gather*}
 for any $v_{1}$ and $v_{2}$. Choosing $v_{1}=u_{1}$ and using
nondegeneracy of $\omega_{N}$ we get $u_{2}=0$ and hence by nondegeneracy of
$\omega_{1}$ we have $u_{1}=0$. It is easy to see now that the diagonal $G$
action on $M\times N$ is Hamiltonian with the moment map $\mu_{M\times
N}:M\times N\rightarrow\mathfrak{k}$ given by
\begin{gather*}
\mu_{M\times N}=\mu_{M}-\mu_{N}.
\end{gather*}

Now let us apply this construction to $N=\Omega_{l}$ -- the coadjoint orbit
of $K$ through $l\in\mathfrak{k}$. Of course, in this particular setting
\begin{gather*}
\mu_{M\times\Omega_{l}}(x,k)=\mu_{M}-k,\,\,\,\, k\in\Omega_{l},
\end{gather*}
 and the fiber of $\mu_{M\times\Omega_{l}}^{-1}(0)$ is given by
\begin{gather*}
\mu_{M\times\Omega_{l}}^{-1}(0)=\{(x,k):\,\mu_{M}(x)=k,\,\,\,
k\in\Omega_{l}\}.
\end{gather*}
 It is easy to see now that
\begin{gather*}
\mu_{M\times\Omega_{l}}^{-1}(0)/K\cong\mu_{M}^{-1}(l)/\mathrm{Stab}(l),
\end{gather*}
 where $\mathrm{Stab}(l)=\{g\in K:\,\mathrm{Ad}_{g}^{\ast}l=l\}$.

The following lemma is the second key ingredient

\begin{lemma}(Kirwan \cite{kirwan84}) Let $K$ be a connected compact
Lie group acting on connected compact K\"{a}hler manifold $(M,\omega)$ by a
Hamiltonian action and let $\mu:M\rightarrow\mathfrak{k}^{\ast}$ be the
moment map. Assume that $x,y\in\mu^{-1}(0)$ and $x\notin K.y$. Then there
exists a $K^{\mathbb{C}}$-invariant open disjoint neighborhoods of $x$ and
$y$ in $M$. \end{lemma}

An obvious consequence of this lemma is

\begin{corollary}\label{almost_homogenous} Assume that $M$ is
$K^{\mathbb{C}}$-almost homogenous, i.e. there is a point $x\in M$ such that
$K^{\mathbb{C}}.x$ is open dense in $M$. Then $\mu^{-1}(0)$ is either empty
or a single orbit of the $K$-action. \end{corollary}

We are now ready to prove Theorem~\ref{brion result}. In order to do this we
need to show that
\begin{gather*}
\mu_{M}^{-1}(M)/\mathrm{Stab}(l),
\end{gather*}
 is a single point for any $l\in\mu(M)$. But by the Marsden-Weinstein
construction this is equivalent to showing that
$\mu_{M\times\Omega_{l}}^{-1}(0)/K$ is a single point. In order to prove it
is enough, by Corollary~\ref{almost_homogenous}, to show that the diagonal
action of $K^{\mathbb{C}}$ has an open dense orbit in $M\times\Omega_{l}$.
Notice first that the coadjoint orbit is the homogenous space
$\Omega_{l}=K^{\mathbb{C}}/P$ where $P=\mathrm{Stab}(l)$ is a parabolic
subgroup. The isotropy subgroup of any point $k\in\Omega_{l}$ is parabolic
group of the form $\mathrm{Stab}(k)=gPg^{-1}$ where $g.l=k$ and $g\in
K^{\mathbb{C}}$. The orbit of diagonal action of $\mathrm{Stab}(k)\subset
K^{\mathbb{C}}$ through $(x,\, k)\in M\times\Omega_{l}$ is
$(\mathcal{O}_{x},\, k)$. For each $k\in\Omega_{l}$ there is a Borel group
$B\subset\mathrm{Stab}(k)$ and since $M$ is spherical for each
$k\in\Omega_{l}$ there is $x_{k}\in M$ such that $\mathcal{O}_{x}$ is open
(they are denoted by dashed lines in the figure 1). Now since
$K^{\mathbb{C}}$ is acting transitively on $\Omega_{l}$ we get by the above
reasoning that $\bigcup_{k\in\Omega_{l}}\mathcal{O}_{x_{k}}$ is open dense
orbit of $K^{\mathbb{C}}$ diagonal action on $M\times\Omega_{l}$. Hence by
Corollary~\ref{almost_homogenous} the quotient
$\mu_{M\times\Omega_{l}}^{-1}(0)/K$ is a single point.
\begin{figure}[H]
\begin{center}
\includegraphics[scale=0.4]{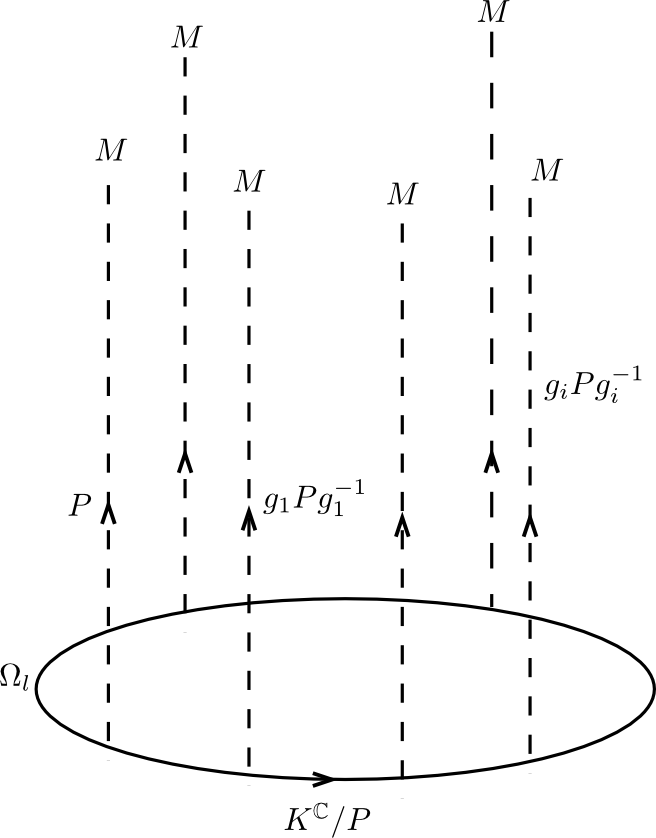}
\caption{The open dense orbit $\bigcup_{k\in\Omega_{l}}\mathcal{O}_{x_{k}}$
of $B$ in $M\times\Omega_{l}$}
\label{fig:fig1}
\end{center}
\end{figure}
In order to emphasize the fact that it is not enough to assume that $M$ is
$K^{\mathbb{C}}$-almost homogenous but we really need the spherical property
in Brion's theorem let us consider the following example

\begin{example}Let $\mathcal{H}_{1}=\mathbb{C}^{2}$, and $\mathcal{H}
=\mathcal{H}_{1}^{\otimes3}$, $K=SU_{2}$,
$K^{\mathbb{C}}=SL_{2}(\mathbb{C})^{\times3}$. Denote by $\{e_{1},\,
e_{2}\}\subset\mathcal{H}_{1}$ an orthogonal basis of
$\mathbb{\mathcal{H}}_{1}$. Consider a state
\begin{gather*}
x=\sqrt{\frac{2}{3}}e_{1}\otimes e_{1}\otimes
e_{1}+\frac{1}{\sqrt{3}}e_{2}\otimes e_{2}\otimes e_{2}.
\end{gather*}

It is the matter of direct calculations to show that $K^{\mathbb{C}}.x$ is
open and dense in $\mathbb{P}(\mathcal{H})$ so $\mathbb{P}(\mathcal{H})$ is
$K^{\mathbb{C}}$-almost homogenous. But from example \ref{ex1} we know that
$\mu$ does not separate all $K$ - orbits.

\end{example}

\noindent It should be also mentioned that in great generality it is known
that every $\mu$-fiber is connected. In the case at hand, due to the fact
that the coadjoint orbits $K.\mu(x)$ are simply-connected, this follows
directly from Brion's theorem. Notice finally that combining Proposition
\ref{open_spherical} with Brion's theorem we get

\begin{corollary}\label{spherical_projective}For $M=\mathbb{P}(\mathcal{H})$,
where $\mathcal{H}$ is the Hilbert space for two fermions, bosons or
distinguishable particles, the corresponding moment map
$\mu:M\rightarrow\mathfrak{k}^{\ast}$ separates all $K$-orbits\end{corollary}

\section{Generic $\mu$ - fibers}
\label{sec:fibers}

Here, in all cases of interest we let $\Sigma$ be the closure of the
$H$-orbit $H.[x_{N}]$, where $x_N=a_N,\, s_N,\, d_N$ in the projective space
$M=\mathbb{P}(\mathcal{H})$. This is itself a complex and therefore
symplectic linear projective subspace of $M$ which is invariant with respect
to the action of the compact torus $T$. Therefore we have its moment map
$\mu_{T}:\Sigma\to\mathfrak{t}^{*}$.

\subsubsection*{Computation of $\mathbf{\mu}$ along $\Sigma$}

Let us begin by computing the restriction $\mu\vert_{\Sigma}$. For this it is
notationally simpler to compute $\mu$ as it is actually defined in the vector
space $\mathcal{H}$. Observe that for any $v\in\mathcal{H}$ if
$[v]\in\Sigma$, then $D_v=C_{v}^{\dagger}C_{v}$ is a non-negative diagonal
matrix with some non-zero entry.

\begin{proposition} If $[v]\in\Sigma$, then
\[
\mu_{\xi}([v])
=\frac{i}{2}\frac{\mathrm{Tr}{(C_{v}^{\dagger}C_{v}\xi)}}
{\mathrm{Tr}{(C_{v}^{\dagger}C_{v})}}
=\frac{i}{2}\frac{\mathrm{Tr}{(D_{v}\xi)}}{\mathrm{Tr}{(D_{v})}}\,.
\]
\end{proposition}

At this point it will be expedient to abandon the parallel treatment of
diagonal and non-diagonal action of $K$ and postpone the latter momentarily.
Although the line of reasoning remains the same, the moment map is slightly
different due to the fact that the Lie algebra
$\mathfrak{k}=\mathfrak{su}_{N}$ for the diagonal actions and
$\mathfrak{k}=\mathfrak{su}_{N}\oplus\mathfrak{su}_{N}$ in the general
two-partite case and it is better to make the concrete calculations
separately.

To treat the cases of diagonal action we decompose the Lie algebra
$\mathfrak{su}_{N}$ as $\mathfrak{t}\oplus\mathfrak{t}^{\perp}$ with respect
to a (unique up to a constant factor) $\mathrm{Ad}$-invariant trace-form. The
elements of $\mathfrak{t}$ are diagonal of the form
$i\mathrm{diag}(\varphi_{1},\ldots,\varphi_{N})$ with $\sum\varphi_{j}=0$. A
basis of $\mathfrak{t}^{\perp}$ is given as follows: For $k<l$ let $E_{kl}$
be the matrix which is one in the i-th row and j-th column, i.e., the matrix
corresponding to the tensor $e_{k}\otimes e_{l}$. The basis of
$\mathfrak{t}^{\perp}$ is then given by the matrices $E_{kl}-E_{kl}^{t}$ and
$i(E_{kl}+E_{kl}^{t})$ with $k<l$. Using this basis one immediately observes
the following fact (see \cite{sawicki11} for detailed calculations)

\begin{proposition} For $[v]\in\Sigma$ it follows that $\mu_{\xi}([v])=0$
for all $\xi\in\mathfrak{t}^{\perp}$.

\end{proposition}

If $[v]\in\Sigma$, for the description of $\mu_{\xi}$ with
$\xi\in\mathfrak{t}$ it is useful to think about diagonal entries of $D_{v}$
as defining a probability measure $P([v])=(p_{1}([v]),\ldots,p_{N}([v]))$,
where
\begin{gather*}
p_{j}([v])=\frac{(D_{v})_{jj}}{\sum_{j=1}^{N}(D_{v})_{jj}},\quad
\sum_{j=1}^{N}p_{j}([v])=1,\,\,\, p_{i}([v])\geq0.
\end{gather*}
\begin{proposition} If $[v]\in\Sigma$ and $\xi=i\mathrm{Diag}(\varphi_{1},\ldots,\varphi_{N})$,
then
\[
\mu_{\xi}([v])=\sum p_{j}([v])\varphi_{j}\,.
\]
\end{proposition}

We may thus regard $\mu_{\xi}([v]))$ as the \emph{expected value} of $\xi$ at
$[v]$.

\subsubsection*{The restricted moment map}

It is now convenient to identify $\mathfrak{k}^{*}$ with $\mathfrak{k}$. We
do so by using the invariant bilinear form given by the trace and regard
moment maps as having values in Lie algebras as opposed to their duals. To
simplify the notation we regard a diagonal matrix
$\xi=i\mathrm{diag}(\varphi_{1},\ldots,\varphi_{N})\in\mathfrak{t}$ as a
vector $(\varphi_{1},\ldots,\varphi_{N})$. The following is then a
translation of the above results.

\begin{proposition} The restriction of the moment map $\mu:M\to\mathfrak{k}$
to $\Sigma$ is the $T$-moment map $\mu^{T}:\Sigma\to\mathfrak{t}$. This
restricted moment map is given by the translated probability vector
\[
\mu([v])=(p_{1}([v]),\ldots,p_{N}([v]))-\frac{1}{N}(1,\ldots,1)\,.
\]
 Its fiber at $[v]\in\Sigma$ is the orbit $T.[v]$. Its image is
the translated probability polyhedron
\[
P=\left\{ \left(p_{1}-\frac{1}{N},\, p_{2}-\frac{1}{N},\ldots,\, p_{N}-\frac{1}{N}\right):\, p_{i}\ge0\ \text{for all}\ i\,\,\mbox{and},\,\sum p_{i}=1\right\} \,.
\]
\end{proposition}

We say that a point $[v]\in\Sigma$ is \emph{generic} if
$p_{1}([v])>\ldots>p_{N}([v])>0$ and let $\Sigma_{gen}$ denote the set of
such points. The following is a direct consequence of the definitions.

\begin{proposition} The restricted moment map has constant rank $n$
on $\Sigma_{gen}$ and maps $\Sigma_{gen}$ onto the interior of the translated
probability polyhedron $P$.
\end{proposition}

\subsubsection*{Orbits of points in $\mathbf{\mu(\Sigma)}$}

Here we let $[v]\in\Sigma$ and analyze the restriction $\mu\vert_{K.[v]}$.
Recall that the $K$-action on $\mathfrak{k}$ is the adjoint representation
which in our matrix setup is just the conjugation. Thus the $K$-isotropy
subgroup at a point $\mu([v])$ is just its centralizer in $K$. We may
therefore restrict our attention to the intersection with the standard
fundamental domain of the Weyl-group.
\[
\mathfrak{t}^{+}:=\{(q_{1},\ldots,q_{N});q_{1}\ge\ldots\ge q_{N},\
\sum q_{i}=0\}\,.
\]
For $\xi\in\mathfrak{t}$ the orbit $K.\xi$ is a flag manifold which depends
only on the equalities among the $q_{i}$. If
\[
q_{1}=\ldots=q_{d_{1}}>q_{d_{1}+1}=\ldots=q_{d_{1}+d_{2}}>\ldots>q_{d_{1}+\ldots+d_{k-1}+1}=\ldots=q_{d_{1}+\ldots+d_{k}}\,,
\]
 then $K.\xi$ is the flag manifold denoted by $F(d_{1},\ldots,d_{k})$.
In that case the isotropy group $K_{\xi}$ is the stabilizer of the standard
flag
\[
0<\mathrm{Span}\{e_{1},\ldots.e_{d_{1}}\}<\mathrm{Span}\{e_{1},\ldots,e_{d_{1}+d_{2}}\}<\mathrm{Span}\{e_{1},\ldots,e_{d_{1}+\ldots+d_{k}}\}=V\,.
\]
The following transversality property will be essential for our computations.

\begin{proposition} If $\xi\in\mathfrak{t}$,
then $T_{\xi}(K.\xi)\cap\mathfrak{t}=\{0\}$.
\end{proposition}

\begin{proof}This is direct consequence of the fact that adjoint
$K$ action does not change the spectrum of antihermitian matrix. More
concretely notice that $g.\xi=\mathrm{Ad}_{g}\xi=g\xi g^{-1}$. So the tangent
space $T_{\xi}(K.\xi)$ is spanned by vectors $[\xi,\,\xi_{1}]$ where
$\xi_{1}\in\mathfrak{k}$. Since $\xi\in\mathfrak{t}$, i.e., $\xi$ is a
diagonal matrix, it is obvious that if $[\xi,\,\xi_{1}]$ is diagonal than it
is equal $0$.\end{proof}

\subsubsection*{Isotropy groups and $\mathbf{\mu}$-fibers at generic points }

For $[v]\in\Sigma_{gen}$ we describe here the isotropy groups $K_{[v]}$ and
$K_{\mu([v])}$. In both cases the $\mu$-image of $\Sigma_{gen}$ is the
interior of the translated probability polyhedron. In the case of bosons this
is open in $\mathfrak{t}$ and thus $K.\mu([v])$ is the full flag manifold
$F(1,1,\ldots,1)$. Thus the isotropy group $K_{\mu([v])}$ is the maximal
torus $T$. Since $\mu$ is equivariant $K_{[v]}$ is contained in
$K_{\mu([v])}=T$. This is the finite group $\Gamma$ consisting of diagonal
matrices $\mathrm{diag}(\pm1,\ldots\pm1)$ with unit determinant.

For fermions the orbit $K.\mu([v])$ is the flag manifold
$F(2,4,\ldots,\ldots,2n-2,2n)$ in the case where $N=2n$ is even and
$\mathcal{F}(2,4,\ldots,2n-2,2n,2n+1)$ in the case where $N=2n+1$ is odd. In
each $2\times2$-block the isotropy group is $\mathrm{U}_{2}$. The determinant
condition then implies that
$K.\mu([v])=\mathrm{S}(\mathrm{U}_{2}\times\ldots\times\mathrm{U}_{2})$ in
the case where $N=2n$ and $K_{\mu([M])}=S(U_{2}\times\ldots\times U_{2}\times
U_{1})$ with $N=2n+1$. In both cases there are $n$ $U_{2}$-factors.

To compute the isotropy subgroup $K_{[v]}$ we therefore only need to compute
the $\mathrm{U}_{2}$-isotropy on $e_{1}\wedge e_{2}$ which is simply
$\mathrm{SU}(2)=\mathrm{USp}(2)$. Thus in both cases, up to possible finite
intersections $K_{[v]}$ is the product
$S^{1}(\mathrm{SU}_{2}\times\ldots\times\mathrm{SU}_{2})$ of $n$-copies of
$\mathrm{SU}_{2}$ with $S^{1}=(e^{i\phi}:\,\phi\in[0,\,2\pi[)$. Let us
summarize these result.

\begin{proposition} Let $[v]\in\Sigma_{gen}$. For bosons $K_{\mu([v])}$
is the maximal torus $T$ and $K_{[v]}$ is the finite group
$\Gamma:=\{\mathrm{Diag}(\pm1,\ldots,\pm1)\in K\}$. For fermions, if $N=2n$
is even,
\[
K_{\mu([v])}=\mathrm{S}(\mathrm{U}(2)\times\ldots\times\mathrm{U}(2))
\]
 and $K_{[v]}$ is the product $\mathrm{SU}_{2}\times\ldots\times\mathrm{SU}_{2}$
of $n$ copies of $\mathrm{SU}_{2}$. If $N=2n+1$, then, modulo the determinant
condition,
\[
K_{\mu([v])}=\mathrm{S}(\mathrm{U}_{2}\times\ldots\times\mathrm{U}_{2}\times\mathrm{U}_{1})
\]
 is the product of $n$ copies of $\mathrm{U}_{2}$ with $\mathrm{U}_{1}$.
In that case
\[
K_{[v]}=S^{1}(\mathrm{SU}_{2}\times\ldots\times\mathrm{SU}_{2})\,.
\]
\end{proposition}

As a result it is a simple matter to describe the fiber
$\mathcal{F}_{[v]}=\mu^{-1}(\mu([v]))$ of the homogeneous fibration
$K_{[v]}\to K_{\mu([v])}$ which by Corollary~\ref{spherical_projective} is
the $\mu$-fiber at $[v]$.

\begin{proposition} For bosons if $[v]\in\Sigma_{gen}$, the
$\mu$-fiber $\mathcal{F}_{[v]}$ can be identified with the maximal torus $T$
modulo the finite group $\Gamma$. In the case of fermions it is the
$(n-1)$-dimensional subtorus
\[
T_{1}=\{\mathrm{diag}(\lambda_{1},\lambda_{1},\ldots,\lambda_{n},\lambda_{n})\in\mathrm{SU}_{N}\}\,.
\]
\end{proposition}

It is in fact not a coincidence that in the case of spherical action the
generic fiber of the moment map is a torus of dimension complementary to that
of the generic $K$-orbit (see \cite{huckleberry90} for this and other results
for such an actions).

\subsubsection*{The slice property}

Using the above results on the properties of the moment map along $\Sigma$,
we will now show that $\Sigma$ has the so called slice property. Namely we
have following general fact

\begin{proposition} \label{slice-1}Let $K$ be a connected compact
group acting smoothly on a connected real analytic manifold $M$ and $N$ be a
closed analytic submanifold. Let $M_{max}$ denote the open subset of points
$x\in M$ so that $K.x$ is of maximal dimension and assume that the complement
of $M_{max}\cap N$ is nowhere dense in $\Sigma$. Furthermore, suppose that if
$x\in M_{max}\cap N$, it follows that $T_{x}K.x+T_{x}N=T_{x}M$. Then $K.N=M$.
\end{proposition}

\begin{proof} By assumption the complement $E$ of $M_{max}\cap N$
is a nowhere dense analytic set. Local linearization of isotropy groups shows
tht $E=K.E$ is at least 2-codimensional and therefore the open complement
$M\setminus K.E$ is connected. The condition $T_{x}K.x+T_{x}N=T_{x}M$ for all
$x\in M_{max}\cap N$ implies that $K.(N\setminus E)$ is open. Since the
complement of $K.N$ is open and $M\setminus K.E$ is connected, the desired
result follows \end{proof}

The closed analytic submanifold $N\subset M$ which has properties given in
the Proposition \ref{slice-1} is called slice. It should be also remarked
that Proposition \ref{slice-1} is not the most general of its type, but that
it is sufficient for proving the slice property in our context.

\begin{theorem}\label{slice} The complex submanifold $\Sigma:=\mathrm{cl}(H.x_{N})$
has the slice property, i.e.\ $K.\Sigma=M$. \end{theorem}

\begin{proof} By Proposition \ref{slice-1}, it is enough to show
the transversality condition is fulfilled for $[v]\in\Sigma_{gen}$. For both
fermions and bosons the $\mu$-fiber at $[v]$ is a torus $T_{0}$ which is
contained in both $\Sigma$ and $K.[v]$. Furthermore,
$\mu\vert_{\Sigma_{gen}}$ has constant rank with image in $\mathfrak{t}$ and
$K.\mu([v])$ is transversal to $\mathfrak{t}$. Thus
\[
\dim(T_{[v]}\Sigma+T_{[v]}K.[v])=\dim T_{[v]}K.[v]+\dim T_{0}=\dim T_{[v]}M\,.
\]
\end{proof}

Using this result it is now possible to give an exact parameterizations of
the $K$-orbits in $M$. For this let $H_{\mathbb{R}}$ be the real points of
$H$ defined as the subgroup corresponding to the Lie algebra $i\mathfrak{t}$.
If
\[
H_{\mathbb{R}}^{+}:=\{\mathrm{diag}(\lambda_{1},\ldots,\lambda_{N}):\,\lambda_{1}\ge\ldots\geq\lambda_{N}>0\}
\]
and $\Sigma_{\mathbb{R}}^{+}:=\mathrm{cl}(H_{\mathbb{R}}^{+}.x_{N})$, then
every $T$-orbit in $\Sigma$ intersects $\Sigma_{\mathbb{R}}^{+}$ in a unique
point. Furthermore, if $[v_{1}],[v_{2}]\in\Sigma_{\mathbb{R}}^{+}$, then
$\mu([v_{1}])=\mu([v_{2}])$ if and only if $[v_{1}]=[v_{2}]$. Since the
$K$-orbits in $\mu(M)$ are exactly parameterized by
$\mu(\Sigma_{\mathbb{R}}^{+})$, the following is immediate.

\begin{corollary}\label{exact slice} Every $K$-orbit in $M$ intersects
$\Sigma_{\mathbb{R}}^{+}$ in exactly one point. \end{corollary}

\section{Orbital symplectic structure}
\label{sec:orbitalstructure}

Here, for an arbitrary point $[v]\in M$ we describe the orbit $K.[v]$ along
with the partial symplectic structure which is a restriction of the
symplectic structure $\omega$ of the projective space
$M=\mathbb{P}(\mathcal{H})$. From Corollary \ref{spherical_projective} we
know that the $\mu$-fiber $\mathcal{F}_{[v]}$ is just the fiber of the
homogeneous fibration $K/K_{[v]}\to K/K_{\mu([v])}$. Thus, in order to
describe $K.[v]$ and its induced partial symplectic struture, it is enough to
compute the isotropy groups $K_{[v]}$ and $K_{\mu([v])}$, thereby describing
the fiber $\mathcal{F}_{[v]}$. By Theorem \ref{slice-1} it is enough to do
this for $[v]\in\Sigma_{\mathbb{R}}^{+}$.

For $[v]\in\Sigma_{\mathbb{R}}^{+}$ the image $\mu([v])$ is the translated
probability vector
\[
\mu([v])=(p_{1}([v]),\ldots,p_{N}([v]))-\frac{1}{N}(1,\ldots,1)=(q_{1},\ldots q_{N})\,.
\]
Since $[v]\in\Sigma_{\mathbb{R}}^{+}$, it follows that $q_{1}\ge\ldots\ge
q_{N}\geq0$. We define the vector $d([v])=(d_{1},\ldots d_{k})$ by the
condition
\begin{equation}
q_{1}=\ldots=q_{d_{1}}>q_{d_{1}+1}=\ldots=q_{d_{1}+d_{2}}>\ldots>q_{d_{1}+\ldots d_{k-1}+1}=\ldots=q_{d_{1}+\ldots d_{k}}\,,\label{eq:d_i}
\end{equation}
where $d_{1}+\ldots+d_{k}=N$. The vector $d([v])$ is of course uniquely
determined by $[v]\in\Sigma_{\mathbb{R}}^{+}$. For example, in the case of
bosons $v=\lambda_{1}e_{1}^{2}+\ldots\lambda_{N}e_{N}^{2}$ with
$\lambda_{1}\ge\ldots\ge\lambda_{N}\ge0$ define the multiplicity vector
$d([v])$. In the case of fermions we are dealing with
$v=\lambda_{1}e_{1}\wedge e_{2}+\lambda_{2}e_{3}\wedge
e_{4}+\ldots+\lambda_{n}e_{2n-1}\wedge e_{2n}$ with
$\lambda_{1}\ge\ldots\ge\lambda_{n}\ge0$. If the equalities for the the
$\lambda_{j}$ define $(\tilde{d}_{1},\ldots,\tilde{d}_{k})$, then
$d_{j}=2\tilde{d}_{j}$ for all $j$. Note that if $N=2n+1$ and
$\mu([v])=(q_{1},\ldots,q_{N})$ it is always the case that $q_{N}=0$ so that
it is always the case $d_{k}\ge1$.

\subsubsection*{Flag manifolds in the $\mathbf{\mu}$-image}

Let $d([v])=(d_{1},\ldots,d_{k})$ be the vector (\ref{eq:d_i}) determined by
$[v]\in\Sigma_{\mathbb{R}}^{+}$.

\begin{proposition}\label{image isotropy} For $[v]\in\Sigma_{\mathbb{R}}^{+}$
the orbit $K.\mu([v])$ is the flag manifold
$F(d_{1},\ldots,d_{k})=K/K_{\mu([v])}$ where
\[
K_{\mu([v])}=\mathrm{S}(\mathrm{U}_{d_{1}}\times\ldots\times\mathrm{U}_{d_{k}})
\]
is the product of unitary groups with the condition of unit determinant.
\end{proposition}

In order to describe the $\mu$-fiber $\mathcal{F}_{[v]}$ we must only
describe the isotropy group $K_{[v]}$.

\subsubsection*{The $\mathbf{\mu}$-fiber along $\mathbf{\Sigma_{\mathbb{R}}^{+}}$}

To simplify the discussion, we disregard the finite groups which arise in the
calculations, e.g., replacing orthogonal groups with special orthogonal
groups. Let us first handle the case of bosons. Restricting to the region
$\Sigma_{\mathbb{R}}^{+}$, we let
\[
v=\sqrt{q_{1}}(e_{1}^{2}+\ldots+e_{d_{1}}^{2})+\ldots+\sqrt{q_{k}}(e_{d_{1}+\ldots+d_{k-1}}^{2}+\ldots+e_{N}^{2})\,,
\]
where $e_{i}^{2}=e_{i}\vee e_{i}$. There are two situations which arise.
First, if $q_{k}\not=0$, i.e., if the associated quadratic form is
nondegenerate, then it is best to express
\[
K_{\mu([v])}=(S^{1})^{k-1}(\mathrm{SU}_{d_{1}}\times\ldots\times\mathrm{SU}_{d_{k}})
\]

Since $K_{[v]}\subset K_{\mu([v])}$, one obtains
\[
K_{[v]}=\mathrm{SO}_{d_{1}}\times\ldots\mathrm{SO}_{d_{k}}\,,
\]
where $SO_{d_{i}}$ are real orthogonal groups. If $q_{k}=0$, then it is
expedient to express
\[
K_{\mu([v])}=(S^{1})^{k-2}(\mathrm{SU}_{d_{1}}\times\ldots\times\mathrm{SU}_{d_{k}-1})\times\mathrm{U}_{d_{k}}\,.
\]
and
\[
K_{[v]}=(\mathrm{SO}_{d_{1}}\times\ldots\times\mathrm{SO}_{d_{k}-1})\times\mathrm{U}_{d_{k}}\,.
\]
For a clean statement, as in our calculations above we replace the fiber
$\mathcal{F}_{[v]}$ by a covering space $\tilde{\mathcal{F}}_{[v]}$. This
differs from $\mathcal{F}_{[v]}$ by a finite group quotient which has been
neglected above and if necessary can be easily computed. For the statement we
denote the symmetric space $\mathrm{SU}_{m}/\mathrm{SO}_{m}$ by $S_{m}$.

\begin{proposition} For bosons and $[v]\in\Sigma_{\mathbb{R}}^{+}$
with multiplicity vector $d=d([v])$ where $v$ is nondegenerate
\[
\tilde{\mathcal{F}}_{[v]}=T_{k-1}(S_{d_{1}}\times\times S_{d_{k}})\,
\]
 where $T_{k-1}$ is a $(k-1)$-dimensional torus. If the tensor $v$
is degenerate, then
\[
\tilde{\mathcal{F}}_{[v]}=T_{k-2}(S_{d_{1}}\times\times S_{d_{k-1}})\,.
\]
\end{proposition}

Now for fermions we choose the point
\[
v=\sqrt{q_{1}}E_{d_{1}}+\ldots+\sqrt{q_{k}}E_{d_{k}},
\]
 where $E_{d_{j}}:=e_{d_{1}+\ldots+d_{j-1}+1}\wedge e_{d_{1}+\ldots+d_{j-1}+2}+\ldots+e_{d_{1}+\ldots+d_{j}-1}\wedge e_{d_{1}+\ldots+d_{j}}$.
As above there are two situations which arise depending on whether or not the
tensor $v$ is degenerate. Note that the case of $N=2n+1$ is always
degenerate. Using the same principles as were applied for the computation in
the case of bosons, we compute $K_{[v]}$. For this we denote the symmetric
space $\mathrm{SU}_{m}/\mathrm{USp}_{m}$ by $A_{m}$.

\begin{proposition} For fermions if $[v]\in\Sigma_{\mathbb{R}}^{+}$
has a multiplicity vector $d([v])=(d_{1},\ldots,d_{k})$, then in the case
when $v$ is nondegenerate
\[
\tilde{\mathcal{F}}_{[v]}=T_{k-1}(A_{1}\times\ldots\times A_{d_{k}})\,.
\]
 In the case when $v$ is degenerate
\[
\tilde{\mathcal{F}}_{[v]}=T_{k-2}(A_{1}\times\ldots\times A_{d_{k-1}})\,.
\]
 \end{proposition}

\subsubsection*{Summary of results}

Omitting the technical descriptions which are given above, we now summarize
our main results.

\begin{theorem} Every orbit $K.[v]$ in $M=\mathbb{P}(\mathcal{H})$
intersects $\Sigma_{\mathbb{R}}^{+}$ in exactly one point. If
$d([v])=(d_{1},\ldots,d_{k})$ is the associated multiplicity vector, then the
image orbit $K.\mu([v])$ is the flag manifold $F(d_{1},\ldots,d_{k})$. The
full image $\mu(M)$ is the union of the orbits $K.\xi$ where $\xi$ is in the
translated probability polyhedron $P-\frac{1}{N}(1,\ldots,1)$. The
$\mu$-fiber $\mathcal{F}_{[v]}$ at $[v]\in M$ is the fiber of the homogeneous
fibration $K/K_{[v]}\to K/K_{\mu([v])}$. Depending on the case and whether or
not tensors $v$ is degenerate, its fiber is the product of a torus and
certain symmetric spaces which are explicitly described above. \end{theorem}
Recall that by Corollary~\ref{spherical_projective}
\[
\dim\mathcal{F}_{[v]}=D([v]),
\]
where $D([v])$ is the degeneracy of the canonical symplectic form on $M$
restricted to the orbit $K.[v]$. For precise computation of dimension of an
orbit and the degeneracy of its induced canonical form it is only necessary
to know the dimensions of the symmetric spaces. For completeness we therefore
remark that
\[
\dim S_{m}=(m^{2}-1)-\frac{m(m+1)}{2}=\frac{(m+1)(m-2)}{2}
\]
and
\[
\dim A_{m}=(m^{2}-1)-\frac{m(m-1)}{2}=\frac{m^{2}+m-2}{2}\,.
\]
One should note that in both cases there is only one symplectic orbit,
$K.[e_{1}\vee e_{1}]$ and $K.[e_{1}\wedge e_{2}]$ in the respective cases.
This is the projectivization of the orbit of the highest weight vector and is
of course complex. The only isotropic orbit is that which is described by
$d=(1,\ldots,1)$ with trivial degeneracy in the cases of bosons and fermions
with $N=2n$, and 1-dimensional degeneracy in the case of fermions with odd
$N$.

\section{Distinguishable particles}

Let us now return to the case of distinguishable particles. As already
mentioned the reasoning is similar up to slight details of concrete
computations. In particular, it is prudent to consider the projectivization
$\Sigma=H.d_{N}$ in $\mathbb{P}(\mathcal{H}_{D})$ of the set of diagonal
matrices in $\mathcal{H}$ as a potential \emph{thick slice} which is
perfectly aligned with respect to $\mu$. The first relevant observation
cocerns $K.\Sigma$.

\begin{proposition}
\[
K.\Sigma=\mathbb{P}(\mathcal{H}_{D})
\]
 \end{proposition} \begin{proof} Suppose $\mathrm{rank}(C_{v})=n$
is maximal. Since we are dealing with the projective space
$\mathbb{P}(\mathcal{H}_{D})$, we may assume that
$C_{v}\in\mathrm{SL}_{N}(\mathbb{C})$ and consider the $K$-orbit of $[v]$. It
is well known that the projectivization of the set of matrices which have
maximal rank is in $K.\Sigma$ as every matrix $C_{v}\in SL_{N}(\mathbb{C})$
can be diagonalized by two unitary matrices \cite{horn85}. Since the set of
states corresponding to maximal rank matrices $C_{v}$ is open and dense in
$\mathcal{H}_{D}$ and $K.\Sigma$ is closed, the result follows. \end{proof}

Now we carry out an argument which is completely analogous to the arguments
in the case of bosons. The formula for the moment map is slightly different
so let us sketch the calculation in a bit of detail. We choose the maximal
toral algebra $\mathfrak{t}_{D}$ to be the direct sum
$\mathfrak{t}\oplus\mathfrak{t}$ of the diagonal matrices in the Lie algebra
$\mathfrak{k}_{D}=\mathfrak{k}\oplus\mathfrak{k}$. Just as before we note
that the elements in $\mu\vert_{\Sigma}$ annihilate the orthogonal complement
$\mathfrak{t}^{\perp}$ of maximal toral algebra. After identifying
$\mathfrak{k}^{*}$ with $\mathfrak{k}$ by the $\mbox{Ad}$-invariant scalar
product given by trace this means
\[
\mu\vert_{\Sigma}:\Sigma\to\mathfrak{t}_{D}\,.
\]
 If $v=\lambda_{1}e_{1}\otimes e_{1}+\ldots+\lambda_{N}e_{N}\otimes e_{N}$,
then for $\xi=i(\varphi_{1},\ldots,\varphi_{n})\oplus i(\psi_{1},
\ldots,\psi_{n})\in\mathfrak{t}$, it follows that
\[
\mu_{\xi}([v])=\sum p_{i}([v])(\varphi_{i}+\psi_{i}),
\]
where the probability distribution vector is given by
$P([v])=(p_{1}([v]),\ldots,p_{N}([v]))$,
\[
p_{i}([v])=\frac{\vert\lambda_{i}\vert^{2}}{\sum_{j=1}^{N}|\lambda_{j}|^{2}}\,.
\]
In order to regard $\mu([v])$ as a vector in
$\mathfrak{t}_{D}=\mathfrak{t}\oplus\mathfrak{t}$ we must translate it:
\[
\mu([v])=(P([v])-\frac{1}{N}I,\, P([v])-\frac{1}{N}I),
\]
 where
\[
I=(1,\ldots,1).
\]

The image $\mu(\Sigma)$ gives all translated finite probability distributions
which can arise. A fundamental region of the Weyl group is given by the image
$\mu(\Sigma^{+})$ where the elements of $\Sigma^{+}$ are of the form
\[
v=\lambda_{1}(e_{1}^{2}+\ldots+e_{d_{1}}^{2})+
\ldots+\lambda_{k}(e_{d_{1}+\ldots+d_{k-1}+1}^{2}+\ldots+e_{d_{1}+
\ldots+d_{k}}^{2}),
\]
where $\lambda_{j}\in\mathbb{R}$ for all $j$ and
$\lambda_{1}>\lambda_{2}>\ldots>\lambda_{k}>0$ and $d_{j}$ denotes the
dimension of the subspace where $v$ has eigenvalue $\lambda_{j}$,
$e_{i}^{2}=e_{i}\otimes e_{i}$ and $\mathrm{rank}(v)=d_{1}+\ldots+d_{k}$.

To complete our analysis of this situation we must describe in detail the
fibration $K/K_{[v]}\to K/K_{\mu([v])}$. Let us begin by computing the
$K$-isotropy group at a generic point $v\in\Sigma_{gen}$, i.e., where
$\lambda_{1}>\lambda_{2}>\ldots\lambda_{n}>0$. Using the fact that
$C_{v}^{\dagger}=C_{v}$ one shows that $K_{[v]}$ is simply
$\Delta(T_{D})=T\times T^{-1}$ which acts by scalar multiplication.

The isotropy group $K_{\mu([v])}$ is in every case the centralizer of
$(P([v])-\frac{1}{N}I,\, P([v])-\frac{1}{N}I)$. Thus, in the generic case it
is just the full maximal torus $T_{D}$. Consequently, the coadjoint orbit is
$K/T_{D}$ and the fiber is the group manifold $T_{D}/\Delta(T_{D})$.

In the general case where the elements of $v$ occur with multiplicities
$d_{1},\ldots,d_{k}$ and $v$ is not necessarily of maximal rank a similar
calculation is made to show that the base $K/K_{\mu([v])}$ is the 2-fold
product $F(d_{1},\ldots,d_{k})\times F(d_{1},\ldots,d_{k})$ of the
corresponding flag manifold. In this case the fiber is the product of the
group manifolds
$(\mathrm{SU}_{d_{i}}\times\mathrm{SU}_{d_{i}})/\mathrm{SU}_{d_{i}}$. It is
perhaps noteworthy that just as in the previous two cases these fibers are
symmetric spaces.

\section{Summary}

We presented the exact description of the partial symplectic structure of all
$K$-orbits in $M=\mathbb{P}(\mathcal{H})$ for two bosons, fermions and
distinguishable particles. To do this we first noticed that in all cases $M$
is in fact a spherical variety, i.e., it is the closure of an open dense
orbit of the Borel group. This observation turns out to be very fruitful.
Namely, by Brion's theorem (\cite{brion87}) it implies that the moment map
$\mu:M\to\mathfrak{k}^{*}$ parameterizes all $K$-orbits in $M$, i.e., $\mu$
maps the set of $K$-orbits in $M$ bijectively onto the set of $K$-orbits in
its image. We gave the exact description of the fibers
$\mathcal{F}_{x}=\mu^{-1}(\mu(x))$ of the moment map. Remarkably in all three
cases these fibers (up to very simple finite-coverings) are products of
certain symmetric spaces.

We believe that the notion of spherical variety should play an important role
also in case of multipartite systems. Of course by simple dimensional
arguments $M$ is typically not spherical in this case but it might happen
that certain of $K^{\mathbb{C}}$-orbits (class of SLOCC states) enjoy this
property. We postpone these problems to the forthcoming publications.

\section*{Acknowledgments}

We gratefully acknowledge the support of SFB/TR12 Symmetries and
Universality in Mesoscopic Systems program of the Deutsche
Forschungsgemeischaft, ERC Grant QOOLAP, a grant of the Polish National
Science Centre under the contract number DEC-2011/01/M/ST2/00379 and Polish MNiSW grant no. IP2011048471.

\section*{Appendix}

\subsection*{Complex symmetric spaces}

Let $G$ be a reductive complex Lie group, i.e., a connected complex Lie group
which is the complexification of some (and therefore any) maximal compact
subgroup $U$.  At the Lie algebra level this means that $\mathfrak {g}$ is
the direct sum $\mathfrak {u}+i\mathfrak {u}$.  Semisimple complex Lie groups
such as the classical groups $\mathrm {SL}_n(\mathbb C)$, $\mathrm
{SO}_n(\mathbb C)$ and $\mathrm {Sp}_{2n}(\mathbb C)$ are the main examples
which occur in our applications. Typical choices for the maximal compact
subgroups $U$ in the classical groups are  $\mathrm{SU}_n$,
$\mathrm{SO}_n(\mathbb R)$ and $\mathrm{USp}_{2n}$. It is important to note
that $G$ can be holomorphically embedded in some $\mathrm {SL}_N(\mathbb C)$
so that $U$ is contained in $\mathrm {SU}_N$.  Thus, $\mathfrak {g}=\mathfrak
{u}+i\mathfrak {u}$ is such that the operators in $\mathfrak {u}$ are
anti-Hermitian and those in $i\mathfrak {u}$ are Hermitian.  The Lie algebra
$\mathfrak {u}$ is a \emph{real form} of $\mathfrak {g}$ in the sense that
there there is an antilinear Lie algebra involution $\sigma :\mathfrak {g}\to
\mathfrak {g}$ with $\mathfrak {u}=\mathrm {Fix}(\sigma )$.  In this case
$\sigma \vert \mathfrak {u}=\mathrm {Id}_{\mathfrak u}$ and $\sigma \vert
i\mathfrak {u}=-\mathrm {Id}_{i\mathfrak {u}}$.

\bigskip\noindent
One often writes $\mathfrak {g}=\mathfrak {g}_u+\mathfrak {p}$ for the
decomposition defined by $\sigma $.  As was noted above, $\mathfrak {p}$ can
be regarded as a subspace of the full space of Hermitian matrices in
$\mathfrak {gl}_N(\mathbb C)$ so that $\mathrm {exp}:\mathfrak
{p}\overset{\cong}{\to} P$ is a diffeomorphism onto a space of
positive-definite Hermitian matrices.  Furthermore, $G$ splits accordingly,
i.e. $G=G_uP$ is a product so that $G/G_u\cong P$ is diffeomorphic to some
$\mathbb R^m$.  Since maximal compact subgroups of simple Lie groups are in
fact maximal and any two are conjugate, it follows that the set $\mathcal C$
of maximal compact subgroups of $G$ can be identified with $G/G_u$ or
equivalently with $P$.

\bigskip\noindent
Complex symmetric spaces which play a role in our investigations are by
definition complex homogenous spaces $G/K$ where $K$ is the fixed point set
of a holomorphic involution $\theta :G\to G$.  Note that $\theta $ induces a
diffeomorphism $\theta :\mathcal C\to \mathcal C$ and recall that no finite
group of diffeomorphism of a cell can act freely.  Therefore $\theta $ has a
fixed point on $\mathcal C$ or equivalently $\theta $ stabilizes some compact
real form which without loss of generality can be taken to be $G_u$.  Let us
consider this matter at the Lie algebra level where we decompose $\mathfrak
{g}=\mathfrak {g}_u+i\mathfrak {g}_u$ into $\sigma $ eigenspaces.  Since
$\theta :\mathfrak {g}\to \mathfrak {g}$ is complex linear and $\theta $
stabilizes $\mathfrak {g}_u$, it is immediate that this splitting is $\theta
$-invariant. Let $\mathfrak {g}_u=\mathfrak {g}_u^++\mathfrak {g}_u^-$ be the
$\theta$-decomposition of $\mathfrak {g}_u$ and note that the $\theta
$-decomposition of $i\mathfrak {g}_u$ is
$$
i\mathfrak {g}_u=(i\mathfrak {g}_u)^++(i\mathfrak {g}_u)^-=
i\mathfrak {g}_u^++i\mathfrak {g}_u^-\,.
$$
Thus
$$
\mathfrak {g}=\mathfrak {g}^++\mathfrak {g}^-=
(\mathfrak {g}^+u+i\mathfrak {g}^+_u)+
(\mathfrak {g}^-_u+i\mathfrak {g}^-_u)\,.
$$
is the $\theta $-splitting of $\mathfrak {g}$.  In particular, since
$\mathfrak {k}=\mathfrak {g}^+$, it follows that the fixed subgroup $K$ is
the complexification of the compact subgroup $K_0$ which is associated to the
Lie algebra $\mathfrak {g}_u^+$.
\begin {proposition} The $\theta $-fixed subgroup $K$
is reductive.
\end {proposition}
\begin {proof}
Since a reductive complex Lie group is by definition one which can be
realized as a complexification of maximal compact subgroup, this follows
immediately from the above discussion.
\end {proof}
\noindent It should be remarked that the quotient of a reductive complex Lie
group by any reductive complex subgroup $H$ can be realized as a closed
$G$-orbit in an appropriately chosen representation space. In particular,
$G/H$ has the structure of an affine variety.  Consequently, in the case at
hand one refers to $X=G/K$ as an \emph{complex affine symmetric space}.

\bigskip\noindent
Observe that $\tau :=\theta \sigma =\sigma \theta $ is an antilinear Lie
algebra involution.  Its fixed point algebra is denoted by $\mathfrak {g}_0$
with associated (noncompact) group $G_0$. Since $\theta $ commutes with $\tau
$, it stabilizes $\mathfrak {g}_0$ and yields a decomposition
$$
\mathfrak {g}_0=\mathfrak {g}_0^++\mathfrak {g}_0^-=
\mathfrak {g}^+_u+i\mathfrak{g}_u^-\,.
$$
This decomposition is written as $\mathfrak {g}_0=\mathfrak {k}_0 + \mathfrak
{p}_0$ and the defining involution, which is the restriction of $\theta $ to
$\mathfrak {g}_0$, is called the \emph{Cartan involution}. Note that since
$\mathfrak {p}_0$ consists of Hermitian matrices, $K_0$ is a maximal compact
subgroup of $G_0$ and $G_0/K_0$ is a symmetric space of noncompact type.
\subsubsection* {Complex symmetric spaces are spherical}
%
The fact that the
symmetric space $X=G/K$ is spherical follows immediately from the Iwasawa
decomposition of $\mathfrak {g}_0$. Starting from the above Cartan
decomposition $\mathfrak {g}_0=\mathfrak {k}_0 + \mathfrak {p}_0$ we now
construct this decomposition.

\bigskip\noindent
First, let $\mathfrak {a}_0$ be a maximal Abelian subspace of $\mathfrak
{p}_0$ and consider its action on $\mathfrak {g}_0$ by the adjoint
representation. Since $\mathfrak {a}_0$ is contained in $\mathfrak {p}_0$
which consists of Hermitian operators in the given realization, the elements
of $\mathfrak {a}_0$ are simultaneously diagonalizable over the reals.  Note
that the $0$-eigenspace of this action, i.e., the centralizer $\mathfrak
{z}_{\mathfrak {g}_0}(\mathfrak {a}_0)$, splits as $\mathfrak {z}_{\mathfrak
{g}_0}(\mathfrak {a}_0)= \mathfrak {m}_0 + \mathfrak \mathfrak {a}_0$ where
$\mathfrak {m}_0\subset \mathfrak {k}_0$.

\bigskip\noindent
An element $\alpha \in \mathfrak {a}_0^*\setminus \{0\}$ is called a
\emph{root} if
$$
\mathfrak {g}_\alpha :=
\{\xi \in \mathfrak {g}_0: x(\xi)=\alpha (x).\xi \ \text{for all}\
x\in \mathfrak {a}_0\}
$$
is nonempty. Let $\Phi $ denote the set of all roots.  For $\alpha \in \Phi $
let $H_\alpha $ be the hyperplane $\{\alpha=0\}$ and define $P$ to be a
component of the complement of $\cup_\alpha H_\alpha $. Fix some $x\in P$ and
say that a root $\alpha \in \Phi$ is positive if $\alpha (x)>0$. Let $\Phi^+$
denote the set of positive roots and
$$
\mathfrak {n}^+_0:=\oplus_{\Phi^+}\mathfrak {g}_\alpha\,.
$$
Finally, observe that the dual action of $\theta $ maps positive roots to
negative roots so that $\Phi =\Phi^-\dot \cup \Phi^+$ and accordingly
$$
\mathfrak {g}_0=\mathfrak {n}^-_0 + \mathfrak {m}_0
+ \mathfrak {a}_0 + \mathfrak {n}^+_0\,.
$$
One checks that $\mathfrak {n}^+_0$ and $\mathfrak {n}^-_0$ are nilpotent Lie
algebras.  Since they are stabilized by $\mathfrak {a}_0$, it is immediate
that $\mathfrak {a}_0 + \mathfrak {n}^+_0$ and $\mathfrak {a}_0 + \mathfrak
{n}^-_0$ are solvable. For simplicity of notation we let $\mathfrak
{n}_0:=\mathfrak {n}_0^+$.
\begin {theorem}(\textbf{Iwasawa decomposition})
$$
\mathfrak {g}_0=
\mathfrak {k}_0 + \mathfrak {a}_0 + \mathfrak {n}_0\,.
$$
\end {theorem}
\begin {proof}
Note that $\theta (\mathfrak {n}^+_0)=\mathfrak {n}^-_0$ so that we have a
decomposition
$$
\mathfrak {n}^-_0 + \mathfrak {n}^+_0=
\delta^- + \delta^+
$$
where $\delta^{\pm}$ is the $\pm 1$-eigenspace of $\theta $. Clearly
$\mathfrak {p}_0=\mathfrak {a}_0 + \delta^-$ and $\mathfrak {k}_0=\delta^+ +
\mathfrak {m}_0$. Since $\delta^+$ (resp.$\delta ^-$) is the diagonal (resp.
antidiagonal) in $\mathfrak {n}_0^+ + \mathfrak {n}_0^-$, the desired result
follows.
\end {proof}
\noindent By complexifying the above Iwasawa decomposition we have the
decomposition
$$
\mathfrak {g}=\mathfrak {k}+\mathfrak {a}+\mathfrak {n}
$$
where $\mathfrak {k}=\mathrm {Fix}(\theta )$ and $\mathfrak {a}+\mathfrak
{n}$ is a solvable Lie subalgebra of $\mathfrak g$.
\begin {corollary}
The homogenous space $X=G/K$ is spherical.
\end {corollary}
\begin {proof}
If $B$ is a Borel subgroup with Lie algebra containing $\mathfrak
{a}+\mathfrak {n}$, then the $B$-orbit of the neutral point in $G/K$ is open.
\end {proof}

\end{document}